\documentclass[9pt]{extarticle} 

\usepackage[numbers, compress]{natbib}
\usepackage[total={5in, 8.75in}]{geometry}

\usepackage{bbm}
\usepackage{amsmath,amssymb,graphicx,amsthm}
\usepackage{mathtools}
\usepackage{xcolor}
\definecolor{mplblue}{HTML}{1f77b4}
\usepackage{hyperref}    
\hypersetup{
    colorlinks=true,
    allcolors = mplblue,
    }
\usepackage{url}    

\newtheorem{thm}{Theorem}
\newtheorem{tlem}{Technical Lemma}

\newtheorem{cor}{Corollary}
\newtheorem{defi}{Definition} 

\renewcommand{\Re}{\mathbb R}
\newcommand{\norm}[1]{\left\lVert#1\right\rVert}
\DeclarePairedDelimiterX{\normm}[1]{\lVert}{\rVert}{#1}

\newcommand{\T}{\mathsf T}
\newcommand{\p}{\theta}
\renewcommand{\r}{x}
\newcommand{\loss}{\ell}
\newcommand{\reg}{\mathtt{Reg}}

\DeclareMathOperator*{\argmin}{arg\,min}

\title{\bf Regret Analysis: a control perspective}

\author{Travis E. Gibson\thanks{Harvard Medical School, Brigham and Women's Hospital, and Broad Institute of MIT and Harvard \texttt{tegibson@bwh.harvard.edu}} \and Sawal Acharya\thanks{Research conducted while at Brigham and Women's Hospital, current address Stanford University}}

\date{}

\begin{document}

\maketitle

\begin{abstract}%
Online learning and model reference adaptive control have many interesting intersections. One area where they differ however is in how the algorithms are analyzed and what objective or metric is used to discriminate ``good'' algorithms from ``bad'' algorithms. In adaptive control there are usually two objectives: 1) prove that all time varying parameters/states of the system are bounded, and 2) that the instantaneous error between the adaptively controlled system and a reference system converges to zero over time (or at least a compact set). For online learning the performance of algorithms is often characterized by the regret the algorithm incurs. Regret is defined as the cumulative loss (cost) over time from the online algorithm minus the cumulative loss (cost) of the single optimal fixed parameter choice in hindsight. Another significant difference between the two areas of research is with regard to the assumptions made in order to obtain said results. Adaptive control makes assumptions about the input-output properties of the control problem and derives solutions for a fixed error model or optimization task. In the online learning literature results are derived for classes of loss functions (i.e. convex) while a priori assuming certain signals are bounded. In this work we discuss these differences in detail through the regret based analysis of gradient descent for convex functions and the control based analysis of a streaming regression problem. We close with a discussion about the newly defined paradigm of online adaptive control.
%In this work we wish to make 3 points. 1) Regret analysis does not give robustness or convergence guarantees 2) Learning rates that asymptotically vanish over time should not be deployed in truly online applications 3) Online gradient descent with learning rates that scale with the features are a robust solution.
\end{abstract}

%\begin{keywords}%
%learning, control, adaptive, adaptive control, optimization, online learning, regret, robust, time varying, online regression, learning rate %
%\end{keywords}

\section{Introduction}
Machine learning and adaptive control have many intersections. As pointed out recently in \citet{gaudio2019connections}, they have common: update laws, robustness modifications (leaky integrator = regularization), both can use projection, have early stopping procedures, adaptive gain/step-size methods, higher order update laws, and the list goes on. The similarities are even more striking in the area of online learning. The purpose of this work is to discuss similarities and differences taken by the adaptive community as compared to the the online learning community both in terms of performance goals and algorithms. We now discuss the origins of optimal control, adaptive control, and adaptive filtering (as they are pertinent to this discussion).\footnote{This short literature review is not meant to be exhaustive} This history is important, because there was a divergence in original literature that seems to have carried on to the present.

This brief historical review begins with the modernization of control theory in the 1960s by Rudolf K\'alm\'an as he  mathematized the state space approach to linear dynamical systems \cite{doi:10.1137/0301010}, introduced algebraic tools to the control field \cite{Kalman1503}, and offered his solution to optimal filtering \cite{10.1115/1.3662552}. Several scientist involved in the development of optimal control (for known dynamics) were also beginning to think about optimal control of uncertain systems. K\'alm\'an's own work in control of uncertain systems can be found in \cite{kalman1958design} where he ``examines the problem of building a machine which adjusts itself automatically to control an arbitrary dynamic process.'' It was originally thought that by the Certainty Equivalence (CE) principle one could iteratively estimate model parameters and then assuming those parameters were ground truth, derive optimal control gains \cite{simon1956dynamic}. Proving that such a scheme resulted in stable dynamics was challenging and CE as a principal didn't gain much traction until some 60 years later \cite{mania2019certainty}.

At the same time others were trying an entirely different approach, which would later become known as model reference adaptive control. In model reference adaptive control a desired reference dynamical system (or reference trajectory) is defined and then control gains are tuned online (in real-time) so as to drive the uncertain dynamical system toward the reference system. The first use of model reference adaptive control is credited to \citet{whitaker1959adaptive} where control parameters were trained in real time using gradient descent. Gradient descent within the context of adaptive control at that time was called the ``MIT rule''.\footnote{The MIT rule in its original form was not a stable algorithm for gradient descent learning with dynamical systems.} Interestingly, just a year later \citet{widrow1960adaptive} showed how to adaptively train a pattern classification machine, using gradient descent as well. The use of gradient descent within the latter context was called the Widrow-Hoff Algorithm, which would later be referred to as the Least Mean Square (LMS) Algorithm \cite{widrow1985adaptive,haykin2002adaptive}. These two publications, just one year apart, using the same principle, kick started model reference adaptive control and machine learning. 
%While the two fields started from the same principle for tuning parameters,  

In model reference adaptive control one has two objectives 1) prove that all time varying parameters/states of the system are bounded, and 2) that the instantaneous error between the adaptively controlled system and a reference system converges to zero over time (or at least a compact set). Importantly, results are derived for arbitrary initial conditions in terms of the state of the plant, and even the initial condition of the plant and controller can be such that the closed loop system is initially unstable. Unfortunately, at the onset of adaptive control heuristics for online gain tuning were used and denoted as ``adaptive'' and in 1967 an experimental X-15 aircraft crashed with such an ``adaptive'' controller \cite{dydek2010adaptive}. Following the crash several decades of work in robust adaptive control followed, resulting in the mature field it is now \cite{narendra1989stable,ioannou1995robust}, where it is even back to flying \cite{lavretsky2013robust,lavretsky2019design}.

In online learning proving boundedness of signals is not the first major concern. Using online regression as an example, the features at time $t$ are not a function of the parameter estimates from the previous time step (where as in adaptive control, the opposite is true). Thus it is not unreasonable to assume in online learning that the inputs to the models are uniformly bounded over time. It is also usually assumed that the gradients of the loss function are uniformly bounded over time. Secondly, the overall goal is not to prove convergence of an error of interest to zero, but instead to focus on optimization (or statistical learning theory) based measures over the length of the learning process. The most prominent optimality measure in online learning is what is referred to as {\it regret}. Regret is defined as the cumulative loss (cost) over time from the online algorithm minus the cumulative loss (cost) of the single optimal fixed parameter choice in hindsight. Its formal definition now follows.
\begin{defi}
Given a sequence of real valued loss functions $\{\loss_t\}$ where $\ell_t: \Re^n \to \Re$ and a bounded convex set $\Theta \subset \Re^n$, the regret of the parameter estimates $\hat \theta_t\in\Re^n$ over time $T$ is defined as.%\vspace{-2pt}
\begin{equation}
\reg(T) \triangleq \textstyle \sum_{t = 1}^T \loss_t(\hat \theta_t) - \textstyle \sum_{t = 1}^T \loss_t(\bar{\theta}), \quad \quad \bar{\theta} = \argmin_{\theta \in \Theta}\textstyle \sum_{t = 1}^T \loss_t(\theta)  \label{eq:reg}
\end{equation}
\end{defi}

The origins of regret analysis for online learning are credited to \citet{blackwell1956} and \citet{hannan1957approximation} where they were interested in optimizing repeated play games. Taking the one player verses nature analogy from \citet{blackwell1956}, at each time $t$ a player chooses its best strategy $\hat \theta_t$ based on past information, and then nature returns the loss $\ell_t(\hat\theta_t)$ after that strategy is played. Note that nature's loss function, $\ell_t$, is time varying (or could be adversarially chosen by a second player at each step). Therefore it is not unreasonable to have accrued losses over time that grow unbounded. Therefore, when considering games like this, regret analysis was proposed so that learning strategies could be compared to the single best fixed strategy in hindsight. Regret analysis has been a part of the online learning community since with regret analysis also appearing in seminal work by \citet{valiant1984theory} and \citet{Littlestone1988} when trying to learn Boolean functions online. For a rigorous treatment of online learning and online convex optimization see \citet{cesa2006prediction,MAL-018,OPT-013}. A brief history of seminal work along the way (that informs our understanding of the field) is in the following footnote.\footnote{Brief history of regret analysis for online learning of convex functions: quadratic cost with linear functions \cite{cesa1996worst}, convex cost with regression, \cite{cesa1999analysis} general convex optimization \cite{Zinkevich:2003:OCP:3041838.3041955}, logarithmic regret for strongly convex functions \cite{hazan2007logarithmic}, and adaptive learning rates \cite{hazan2008adaptive,duchi2011adaptive}}

The purpose of this manuscript is to point out the differences in the analysis techniques and learning rate scalings used when analyzing gradient descent within the two communities. From the online learning community there has been great strides to solve general (even adversarial) problems. The control community is more rigid in their problem statements. That being said, the learning rates and analysis proposed within the control community give convergence guarantees. The main distinguishing feature between the fields is the use of learning rates in adaptive control that scale with the features in a regression problem. The online learning community has instead used learning rates that vanish over time. %We have questions about the practicality of systems with vanishing learning rates and discuss this in detail in the conclusions section.

In subsequent sections the following three flavors of gradient descent will be used
\begin{align}
    \hat\theta_{t+1} & = \hat \theta_t - \eta_t \nabla  \loss_t(\hat\theta_t) \label{g:1} \\
    \hat\theta_{t+1} & = \mathtt{Proj}_\Theta(\hat \theta_t - \eta_t \nabla \loss_t(\hat\theta_t)) \label{g:2} \\
    \hat\theta_{t+1} & = \hat \theta_t - \eta_t \nabla \loss_t(\hat\theta_t) - \sigma_t \hat\theta_t \label{g:3}
\end{align}
where $\hat\theta_t\in\Re^n$ is the parameter estimate at time $t\in \mathbb N_{>0}$, $\eta_t>0$ is the time varying learning rate, $\ell_t\in \Re$ is the loss function, $\nabla\ell_t(\hat\theta_t)$ is the gradient of the loss function with respect to the parameter estimate, and $\sigma_t\geq 0$ is a leaky integration modification for robustness purposes (to be defined more explicitly at a later time). The operator $\mathtt{Proj}_\Theta$ projects arguments onto the convex set $\Theta$.\footnote{$\mathtt{Proj}_{\mathcal Y} (z) = \argmin_{y\in \mathcal Y} \norm{y-z}$}. Slightly abusing notation we will make the following definition moving forward 
\begin{equation*}
    \nabla \loss_t:=\nabla \loss_t(\hat\theta_t).
\end{equation*} 
The paper is organized as follows: In Section 2 we cover the basics of regret analysis in online learning. In section 3 we give the adaptive control based analysis to online (real-time streaming) regression. In Section 4 we review the results presented and close with some concerns and possible ideas for closing the gap between the two approaches.

%There are many applications where decisions have to be made in an online or sequential manner. Control theory was established to address how inputs to a dynamical system affect its stability or robustness, and Online Learning (OL), a sub-field of Machine Learning (ML), was developed to address learning as it relates to sequentially optimizing a potentially time varying objective function. One of the original motivating applications for OL was investment portfolio management. At discrete time points an investor decides how to proportion money within a portfolio. Classic models used for regression or classification can also be turned into an ``online'' algorithm by treating the updates as sequential decisions or by treating the data as a streaming signal (this has intimate connection to batch optimization, the core algorithm for training neural networks). One natural marriage for the field is adaptive control, where control parameters are updated continuously online.

%Wanting to get sub-linear regret bounds has resulted in choosing of learning rates that either decay over time, or are a function of the length of the expected learning application. In both cases the learning rate tends to zero as time goes to infinity. For an online application this does not make much since, and appears to be mainly driven by the proof techniques and does not comport with the actual goals of the application.

%Within control some of the main motivating applications are with regard to physical systems, i.e. an airplane.

\section{The basics of regret analysis for online convex optimization}
The learning rates and corresponding regret bounds for online convex optimization are as follows \cite{hazan2019lecture}, \cite{levy17}:%
\begin{align*}
  %& \text{Regression} &    \eta& \propto \frac{1}{\sqrt T} &     \reg(T) &= O(\sqrt T) \\
 & \text{Convex} &     \eta_t& \propto \frac{1}{\sqrt {t}} \quad \text{or} \quad \eta_t \propto \frac{1}{\sqrt{\sum_{\tau=1}^t \norm {\nabla \ell_\tau}^2}}&     \reg(T) &= O(\sqrt T)  \\
 & \text{Strongly\ Convex} &     \eta_t& \propto \frac{1}{{t}} &     \reg(T) &= O(\log T). 
\end{align*}
It is important to note that the regret bounds in both cases are sub-linear. Thus for these two cases the average regret over time is zero, $\lim_{T \to \infty}\frac{\reg(T)}{T}=0$. Another distinctive feature of these results is that the learning rates either tend to zero or are monotonically decreasing over time. We will now present a key technical lemma that is widely used in the analysis of online convex optimization.

%When analyzing optimization error, we usually consider the average regret (\cite{hazan2019lecture} section 4.2). Let $\bar{x}_T = \dfrac{1}{T} \sum_{t =1}^T x_t$. By Jensen's inequality, $f(\bar{x}_T) \leq \dfrac{1}{T}\sum_{t = 1}^T f(x_t) \Rightarrow f(\bar{x}_T) - f(x^*) \leq \dfrac{1}{T} \left (\sum_{t = 1}^T f(x_t) - f(x^*) \right)$. Hence, 
%\begin{equation}\label{ineq1}
   % f(\bar{x}_T) \leq f(x^*) + \dfrac{R_A(T)}{T}
%\end{equation}
%\textbf{Why use regret }Let the data points $(x_i, y_i)$, where $x_i \in \mathbbm{R}^d, $ be sampled from some distribution space $\mathcal{D} = X \times Y$. Let $\mathcal{H}$ be the hypothesis space of functions $f: X \rightarrow Y$, $f:\mathbbm{R} \times \mathbbm{R} \rightarrow \mathbbm{R}$ be the loss function. \\

%Taking expectation on \eqref{ineq1}, $$\mathbbm{E}_{(x_i, y_i) \sim \mathcal{D}}[ f(\bar{x}_T)] \leq \mathbbm{E}_{(x_i, y_i) \sim \mathcal{D}}[ f(x^*)] + \mathbbm{E}_{(x_i, y_i) \sim \mathcal{D}} \left[\dfrac{R_T(A)}{T}\right] = \mathbbm{E}_{(x_i, y_i) \sim \mathcal{D}}[ f(x^*)] + \dfrac{R_T(A)}{T} $$
%The LHS of the above inequality is the generalization error. Thus, regret For the generalization error to be as small as possible, the second term in the right hand side must be close to 0. This happens when the regret is $\mathcal{O}(\sqrt{T}), \mathcal{O}(1),  \mathcal{O}(\log T) $ etc.

\begin{tlem}\label{tl1}
For $t\in\mathbb N_{> 0}$  the following bound holds
$  \sum_{t=1}^T \frac{1}{\sqrt{t}} \leq 2\sqrt T.$
\end{tlem}
\begin{proof}
From definition of the Riemann Integral, and the monotonicity of $\frac{1}{\sqrt t}$, the following bound holds $ \sum_{t=1}^T \frac{1}{\sqrt{t}} \leq \int_0^T  \frac{1}{\sqrt{\tau}} \, \mathrm d \tau$. Computing the integral gives the right hand side bound above. 
\end{proof}
With this lemma we are now ready to prove the sublinear regret bound stated above for convex online learning with the vanishing learning rate $\eta_t \propto 1/\sqrt{t}$. 
\begin{thm}
Let $\Theta$ be a bounded convex set with diameter D, and $\{ \ell_t\}$ be a sequence of convex functions such that $\norm{\nabla \ell_t} \leq G$. Then with learning rate $\eta_t = \tfrac{D}{G\sqrt{t}}$ the regret of the algorithm in Equation \eqref{g:2} can be bounded as follows $\reg(T)=O(\sqrt{T})$ \end{thm}
\begin{proof} (following \cite{hazan2019lecture})
We begin by deriving a bound between the one time point step ahead parameter estimate $\hat \theta_{t+1}$ and the optimal fixed parameter value in hindsight, $\bar \theta$, defined in \eqref{eq:reg} as follows
$$\normm[\big]{ \hat{\theta}_{t+1} - \bar{\theta}}^2 = \normm[\big]{\mathtt{Proj}_\Theta(\hat{\theta}_t - \eta_t\nabla \ell_t) - \bar{\theta}}^2\leq \normm[\big]{\hat{\theta}_t - \eta_t\nabla \ell_t - \bar{\theta}}^2$$
where we have used the update law for the parameter estimate and properties of projection operators. Expanding the squared norm and recollecting terms
$$ \normm[\big]{\hat{\theta}_{t+1} - \bar{\theta}}^2 \leq \normm[\big]{\hat{\theta}_t - \bar{\theta}}^2 + \eta_t^2\normm[\big]{\nabla \ell_t}^2 - 2\eta_t\nabla \ell_t^\T(\hat{\theta}_t - \bar{\theta}) $$ and by swapping terms on the left and right hand side
$$2\nabla \ell_t^\T(\hat{\theta}_t - \bar{\theta}) \leq \frac{\normm[\big]{\hat{\theta}_t - \bar{\theta}}^2}{\eta_t} - \dfrac{\normm[\big]{\hat{\theta}_{t+1} - \bar{\theta}}^2}{\eta_t} + \eta_t \normm[\big]{\nabla \ell_t}^2. $$
Next we wish to bound some terms via a telescoping sequence, and by defining $\dfrac{1}{\eta_0} = 0$ and summing we have 
\begin{equation*}
\sum_{t = 1}^T2\nabla \ell_t^\T(\hat{\theta}_t - \bar{\theta}) \leq \sum_{t = 1}^T \normm[\big]{\hat{\theta}_t - \bar{\theta}}^2 \left( \dfrac{1}{\eta_t} - \dfrac{1}{\eta_{t -1}}\right) + \sum_{t = 1}^T \eta_t\norm{\nabla \ell_t}^2.
\end{equation*}
Using the fact that  $\ell_t$ is convex and $\normm[\big]{\hat{\theta}_t - \bar{\theta}} \leq D$
$$2 \sum_{t =1}^T \left( \ell_t(\hat{\theta}_t) - \ell_t(\bar{\theta}) \right)\leq D^2 \dfrac{1}{\eta_T} + G^2\sum_{t = 1}^T \eta_t.$$
This is the crucial part of the proof, and where one recognizes the need to have a learning rate that balances the trade off between $1/\eta_T$ and $\sum_{t=1}^T \eta_t$. Substituting our definition of the learning rate and applying Technical Lemma \ref{tl1} we have that $D^2 \dfrac{1}{\eta_T} + G^2\sum_{t = 1}^T \eta_t \leq 3GD\sqrt{T}$
%From our definition of $\eta_t$ 
%$$2 \sum_{t =1}^T \left( \ell_t(\hat{\theta}_t) - \ell_t(\bar{\theta}) \right)\leq DG\sqrt{T} + DG\sum_{t =1}^T\sqrt{t}$$
%Using lemma 2 to bound the second term of RHS, 
%$$\sum_{t =1}^T \left( \ell_t(\hat{\theta}_t) - \ell_t(\bar{\theta}) \right)\leq \dfrac{3}{2}GD\sqrt{T}$$
\end{proof}

We will reiterate again that at the crux of the proof that online convex optimization with gradient descent attains sub-linear regret is being able to balance the two terms $\frac{1}{\eta_T}$ and $\sum \eta_t$. Another frequently used learning rate that can perform this balancing act is $\eta_t \propto \frac{1}{\sqrt{\sum_{\tau=1}^t \norm{\nabla \ell_\tau}^2}}$. Methods deploying this type of learning rate have been referred to as adaptive subgradient algorithms and have become wildly popular for training deep neural networks (\cite{duchi2011adaptive}). With the following lemma and corollary one is able to show $\sum_{t=1}^T \frac{\norm{\nabla \ell_t}^2}{\sqrt{\sum_{\tau=1}^t \norm{\nabla \ell_\tau}^2}}=O(\sqrt{T})$ which in turn can be used to prove  $O(\sqrt{T})$ regret. 
\begin{tlem} (\cite{DBLP:journals/corr/abs-1002-4908})
Let $\{b_t\}$ be a sequence of elements in $\Re_{\geq 0}$ indexed by $t\in\mathbb N_{>0}$. Then the following bound holds $\sum_{t=1}^T \frac{b_t}{\sqrt{\sum_{\tau=1}^t {b_\tau}}} \leq 2\sqrt{\sum_{t=1}^T b_t}$. 
\end{tlem}
%\begin{proof} (following \cite{DBLP:journals/corr/abs-1002-4908}) The proof is by induction. For $t = 1$ it follows that $\dfrac{b_1}{\sqrt{b_1}} = \sqrt{b_1} \leq 2\sqrt{b_1}$. Assuming the inequality holds for $t=T-1$ the bound for $t=T$ is as follows
%\begin{equation}\label{ineq2} 
%    \begin{split}
%        \sum_{t =1}^{T-1}\dfrac{b_t}{\sqrt{\sum_{\tau = 1}^{t}}b_\tau} + \dfrac{b_T}{\sqrt{\sum_{\tau = 1}^T b_\tau}} &\leq 2\sqrt{\sum_{t = 1}^{T - 1}b_t} + \dfrac{b_T}{\sqrt{\sum_{\tau = 1}^T b_\tau}}  \\ &= 2\sqrt{B - b_T} + \dfrac{b_T}{\sqrt{B}}
%   \end{split}
%\end{equation}
%The RHS of \eqref{ineq2} is maximized when $-\frac{1}{\sqrt{B-b_T}} + \frac 1 {\sqrt{B}}=0$, found via differentiating with respect to $b_T$, and this occurs when $b_T=0$. Setting $b_T=0$ on the RHS of \eqref{ineq1} gives the desired bound. \end{proof}
\begin{cor}
Let $\{a_t\}$ be a \underline{bounded} sequence of elements in $\Re$, thus there exists a constant $D>0$ such that $\norm{a_t}\leq D$. The following bound holds $
    \sum_{t=1}^T \frac{a^2_t}{\sqrt{\sum_{\tau=1}^t {a^2_\tau}}} \leq 2 D \sqrt{T}$.
\end{cor}
We will not discuss online learning with strongly convex functions, but just to illustrate why a learning rate of $1/t$ might attain a $\log(T)$ regret bound, just remember that the integral of $1/t$ is $\log(t)$ plus a constant.

\section{Online regression using analysis from the control community} \label{sec:control}

We now present online regression using the same analysis tools that are used to prove stability for the (model reference) adaptive control of discrete time systems. At the outset we have to note that the tools just presented are non-trivially more generally applicable. In this section we will have to detail most aspects of the model and define an explicit loss function so as to perform analysis. First we present the analysis of online regression and then we present the analysis of online regression in the presence of perturbations.

Consider the following regression problem
\begin{align}
    y_t  &= \theta_*^{\T} x_t  \label{eq:regy}\\
    \hat y_t &= \hat\theta_t^\mathsf T x_t  \label{eq:regyh}
\end{align}
where $x_t\in \Re^n$ are bounded time varying features,\footnote{In the regular adaptive control setting the features are the states of the plant and are not assumed to be bounded, but are proved to be bounded.} $\theta_*\in \Theta \subset \Re^n$ is an unknown parameter, $\hat\theta_t\in \Re^n$ is the online parameter estimate, $y_t$ is the known scalar ``reference'' output and $\hat y_t$ is the scalar time varying estimate of $y_t$. With the output error defined as $e_t\triangleq \hat y_t-y_t$ our loss function for this section will be defined as $\ell_t = \frac{1}{2}e_t^2$. Finally we will use the following variable $\tilde \theta_t \triangleq \hat \theta_t-\theta_* $ to denote the difference between the parameter estimate and the unknown parameter. Unlike before, in this section we will be using a learning rate
\begin{equation}\label{adlearn}
    \eta_t = \frac{\alpha_t}{m+x_t^\mathsf T x_t}
\end{equation}
where $\alpha_t, m>0$. Note that this learning rate scales with the inverse of the features and does not monotonically decrease or tend to zero over time.

%Online regression problem 1. Fixed strategy
%Online regression problem $y_t = x^{*\mathsf T} \theta_t$, we have an estimate of that $\hat y_t = x_t^\mathsf T \theta_t$, then the error is Defined as  $e_t= \tilde x_t ^\mathsf T \theta_t$, where where $\tilde x_t = x_t-x^*$ the loss function at each time is then $\ell(x_t, \theta_t) = \frac{1}{2}e^2_t$

\begin{thm}\label{thmad1}
For the online regression problem in \eqref{eq:regy} and \eqref{eq:regyh} with gradient descent algorithm in \eqref{g:1} with learning rate $\eta_t$ defined in \eqref{adlearn}, $\alpha_t \in (0,2)$ and $m>0$ all signals are uniformly bounded $\sum_{t=1}^\infty (\hat y_t - y_t)^2$ is bounded, and furthermore, $\lim_{t\to \infty} \hat y_t = y_t$.
\end{thm}
\begin{proof} (following \citet{goodwin1980discrete})
Define a Lyapunov candidate $V_t = \tilde \theta_t^\mathsf T \tilde \theta_t $  (note that this is not the loss function at time $t$, 
$\tfrac{1}{2}e_t^2$). 
The one time discrete difference of the candidate function is as follows $ \Delta V_t \triangleq V_{t+1}- V_t =  \tilde \theta_{t+1}^\mathsf T \tilde \theta_{t+1} - \tilde \theta_t^{\mathsf T} \tilde \theta_t $. Substituting the update law from \eqref{g:1} results in $$ \Delta V_t =  [\hat\theta_{t} - \eta_t \nabla \ell_t -\theta_*] ^\mathsf T [ \hat\theta_{t} - \eta_t \nabla \ell_t-\theta_*] - \tilde \theta_t^{\mathsf T} \tilde \theta_t.$$Collecting and canceling $\tilde x_t$ and substituting the gradient of our loss function 
%$$
%\Delta V_t =  \tilde x_t^{\mathsf T} \tilde x_t  - 2 \tilde x_t^\mathsf T \eta \nabla \ell_t + \norm {\eta \nabla \ell_t}^2  - \tilde x_t^{\mathsf T} \tilde x_t
%$$
%Cancelling terms and expanding the gradient, we have ...... ohh... nm you are right :).....
$$
\Delta V_t = - 2 \tilde \theta_t^\mathsf T \eta_t e_t x_t + \eta_t^2 e_t^2 x_t^\mathsf T x_t.
$$
Substituting the learning rate for $\eta_t$ from \eqref{adlearn} and collecting the error terms, $e_t$, the following bound holds
%$$
%- 2 \frac{(\tilde x_t^\mathsf T \theta_t)^2}{{1+\theta_t^\mathsf T w_t}}  + \frac{\alpha^2 (\tilde x_t^\mathsf T \theta_t)^2 \theta_t^\mathsf T \theta_t}{{(1+\theta_t^\mathsf T \theta_t)^2}} % 
%$$ Separating terms for ease of noting bound
\begin{equation}
\Delta V_t = \alpha_t \left[- 2   + \alpha_t \frac{ x_t^\mathsf T x_t}{m+x_t^\mathsf T x_t} \right] \frac{e_t^2}{{m+x_t^\mathsf T x_t}} \leq 0 \quad \quad (\forall\ 0\leq \alpha_t\leq 2).
\end{equation}
$\{V_t\}$ is a bounded monotonic sequence and thus it converges. This then implies that $\hat\theta_t$ is bounded (without having to have used projection). Then in the final step we show that $\sum e_t^2$ is uniformly bounded. Summing both sides and multiplying by a negative one it follows that
\begin{equation}
\begin{split}
    V_{1} - V_T &= \sum_{t=1}^T   \alpha_t \left[2   - \alpha_t \frac{ x_t^\mathsf T x_t}{m+x_t^\mathsf T x_t} \right] \frac{e_t^2}{{m+x_t^\mathsf T x_t}} \\
&\geq  \sum_{t=1}^T   \alpha_t \left[2   - \alpha_t \right] \frac{e_t^2}{{m+x_t^\mathsf T x_t}}
\end{split}
\end{equation}
Thus  $\sum_{t_0}^T {e_t^2} \leq c_1 V_{1}$, where $c_1>0$ is a constant which is a function of the known upper bound on the streaming regressor vectors and the user defined $\alpha_t$ . Taking the limit as $t$ tends to $\infty$, $e_t \in \ell^2 $ and because $\theta_t,x_t\in \ell^\infty$ this implies that $e_t\in \ell^\infty$ as well. Combing these results we have that $e_t\in \ell^2 \cap \ell^\infty $ and therefore $\lim_{t\to\infty} e_t=0$.\end{proof}

%\begin{cor} For the system in Theorem %\ref{thmad1} $\reg(T)=O(1)$
%\end{cor}

The crux to this proof is ensuring that we have a learning rate that is smaller than the squared 2-norm of the features at any instance in time. This kind of scaling also appears in what is called Normalized LMS, but we were unable to find any of the literature stemming from \cite{widrow1985adaptive,haykin2002adaptive} that had this style of proof with these strong guarantees. It is also important to note that we did not prove that $\hat\theta_t$ tends to $\theta_*$ but we were still able to prove that $\hat y_t$ asymptotically converges to $y_t$. In order to prove that $\hat\theta_t$ converges to $\theta_*$ extra conditions are required on the regressor vector $x_t$, namely that $x_t$ is Persistently Exciting (PE) (\cite{doi:10.1137/15M1047805}).\footnote{The vector $x_t\in \Re^n$ is PE if there exists $\beta,T>0$  such that $\sum_{\tau = t}^{T+t} x_\tau x_\tau^\T \succeq \beta I \quad \forall\ t.$} Having the regressor vector PE is a necessary and sufficient condition for exponential stability in the large for the point $\tilde\theta_t = 0.$

We concede that even with such strong guarantees the model is ``simple'' by ML standards and the reference model is not adversarial either. Within the adaptive control literature the closest things we have to adversarial types of proofs are when stability is proved in the presence of unmodelled dynamics or bounded disturbances. This line of work has lead to several robustness modifications \cite{ioannou1995robust,narendra1989stable} and we present a flavor of that kind of model and learning rule modification here now.

Consider the disturbed regression model
\begin{equation}
    y_t  = \theta_*^{\T} x_t +d_t \label{eq:regd}
\end{equation}
where $d_t$ is now a bounded time varying disturbance and  $\theta_* \in \Theta$ for some known bounded convex set $\Theta$. In order to prove that we can stably adapt to the disturbed reference we will need to modify the standard gradient descent law and thus we will be using the parameter update in \eqref{g:3} where the $\sigma$-modification (\cite{ioannou1995robust}) term $\sigma_t$ is defined as 
\begin{equation}\label{sig}
\quad \sigma_t =  \begin{cases} 0 & \text{ if } \hat\theta_t \in \Theta \\ \sigma>0 & \text{ otherwise }.\end{cases}
\end{equation}
Another name for such a modification is ``leaky'' integrator and it is important to note that the same type of term $\sigma_t \hat\theta_t$ arises when one performs gradient descent on an L2 regularized loss function (i.e. $\nabla_\theta \norm \theta^2 \propto \theta$). We now prove boundedness of the parameter estimates in the presence of disturbances.

\begin{thm} For the online regression problem in \eqref{eq:regd} and \eqref{eq:regyh} with gradient descent algorithm in \eqref{g:3}, bounded disturbance $\norm{d_t} \leq d$, with learning rate $\eta_t$ defined in \eqref{adlearn}, $\alpha_t \in (0,2)$, $\sigma_t$ defined as in \eqref{sig} and $m>\max\{d,1\}$, all signals are uniformly bounded.
\end{thm}
\begin{proof} (following \cite{ioannou1986robust})
Define a Lyapunov candidate $V_t = \tilde \theta_t^\mathsf T \tilde \theta_t $ as before and 
 $ \Delta V_t \triangleq V_{t+1}- V_t =  \tilde \theta_{t+1}^\mathsf T \tilde \theta_{t+1} - \tilde \theta_t^{\mathsf T} \tilde \theta_t $. Evaluating $\Delta V_t$ with the gradient update in \eqref{g:3} and after completing the square the following bound holds
\begin{equation}
\Delta V_t \leq \dfrac{-\alpha}{2}\dfrac{(\tilde \theta_t^T x_t)^2}{m + x^\T x} - \dfrac{1}{4}\sigma_t \hat\theta_t^\T\tilde \theta_t + c_2 d^2
\end{equation}
where $c_2>0$ is a constant. At this point it may appear that we are stuck because of the sign indefinite potentially unbounded term $\hat\theta_t^\T \tilde \theta_t$. It turns out however that this term is not sign indefinite. By the definition of $\sigma_t$ in \eqref{sig}, $\sigma_t$ is nonzero only when $\hat\theta_t^\T \tilde \theta_t>0$ and thus $-\sigma_t \hat\theta_t^\T\tilde \theta_t \leq 0$. Thus it follows that $\Delta V_t \leq 0$ when $V_t$ is sufficiently large. This implies that $\tilde\theta_t$ is uniformly bounded. Then given that $x_t$ is bounded it also follows that $e_t$ is bounded as well. \end{proof}

\section{Contrasting the two approaches and the new paradigm of Online Adaptive Control}\label{sec:oac}

%Online Learning for linear for regression and convex optimization
%\cite{,cesa2006prediction,hazan2007logarithmic,hazan2008adaptive,DBLP:journals/corr/abs-1002-4908,duchi2011adaptive,MAL-018,OPT-013,hazan2019lecture}

So far we have presented the basics of regret analysis for online convex optimization, and stability proofs for streaming real-time regression that are inspired by the model reference adaptive control literature. There are two main differences in the proof structure that are subtle but important. In regret analysis, one tries to directly bound the loss function and the accumulated cost associated with it. In the adaptive results just presented, however, bounding the parameter error (i.e. $\tilde \theta$) is what is tackled first. In the process of bounding the parameter error the model following error (which defines our loss in Section \ref{sec:control} ends up being square integrable bounded ($e \in \ell^2$). For online learning it never made sense to first bound the parameter error, because if one uses projection you get this for free. This very slight difference in initial goals results in the two fields using two different principles to define learning rates. The online learning community tries to ensure $\sum \eta_t$ does not grow too quickly while also ensuring that $\eta_t$ decays gracefully (not too fast). In the control inspired stability proofs one has to ensure that the term $\eta_t x^\T_t x_t$ is sufficiently small, thus leading to the learning rate scaling $\eta_t \propto  \frac{1}{m+ x^\T_t x_t}$.

One interesting intersection of control and online learning is the newly minted area of ``Online Adaptive Control'' (OAC), where we have now come full circle and arrived back at the original paradigm of CE control. %Our only goal in this section is to make the connection to our prior discussion on regret based analysis and the need for vanishing ``learning rates''. 
For the purposes of this discussion we will be considering the linear dynamics
\begin{equation*}
    \label{eq:linear}
    x_{t+1} = A_* x_t + B_* u_t + w_t
\end{equation*}
where $w_t$ is the process disturbance and $(A_*,B_*)$ are unknown parameters. In OAC one tries to minimize the regret associated with the quadratic cost
\begin{equation*}
\ell_t = x_t^\T Q x_t + u_t^\T R u_t, \quad Q\succeq 0, R\succ 0.
\end{equation*}
If the dynamics were known then the input that minimizes the infinite time horizon cost  $\sum_{t=0}^\infty\ell_t$ would simply be the linear feedback $u_t=Kx_t$ where $K$ is the solution to the Discrete Algebraic Riccati Equation (DARE). In the setting where the plant is unknown however, the regret optimal input under the OAC paradigm takes the form
\begin{equation*}
    u_t=K_tx_t+\xi_t.
\end{equation*}
where in addition to the time varying linear feedback $K_tx_t$ we have an exploration signal
\begin{equation*}
    \xi_t\sim \texttt{Normal}(0,\sigma_{\xi,t}^2).
\end{equation*}
Based on the CE principle, at each time $t$, one estimates the unknown model parameters $(A_*,B_*)$ giving ourselves $(\hat A_t, \hat B_t)$ and then using a closed form solution to the DARE one can define a map
$(\hat A_t, \hat B_t)\mapsto K_t$.

Now this all straight forward but in the previous section we emphasized that all regret based update laws had this interesting property of vanishing learning rates. Here, none of the parameters are updated with gradient descent and thus there is no learning rate. We do still however have a hyper-parameter in this paradigm that directly controls how much learning is occurring and that is the user specified exploration signal variance $\sigma_{\xi,t}^2$. In order to be regret optimal the exploration signal must vanish, with one popular choice being\cite{mania2019certainty}
\begin{equation*}
    \sigma_{\xi,t}^2 \propto \frac{1}{\sqrt{t}},
\end{equation*}
which should now be familiar.

This was just one example. There are other flavors of regret based optimal control whereby the exploration signal does not decay over time, but instead the control parameters themselves are updated via gradient descent with a time-varying and asymptomatically decreasing learning rate \cite{simchowitz2020improper,hazan2020nonstochastic}. Either directly or indirectly to be regret optimal the rate at which you learn must decrease over time.

\section{Closing}
We hope this note was helpful in familiarizing the control community with some of the technical nuances of regret analysis.

%, that is, unless you are OK with algorithmic induced turbulence on your next flight.

%In the adaptive control problem the regret is with respect to a player or reference with a constant $x^*$. In the online learning setting the reference information as we would call it in adaptive control could be adversarially changing over time. That is a regression parameter from the reference could be time varying $x^\prime_t$

%This is more challenging because earlier we exploited the fact that $x^*$ is fixed and it allowed us to cancel terms after substitution of the update law. Next steps would be to see if we could get a regret like bound... so now let $x^*$ be the best fixed policy, ... then 

% Acknowledgments---Will not appear in anonymized version
%\acks{We thank Georg Gerber at HMS, and Anuradha Annaswamy and Joseph Gaudio at MIT for fruitful discussions on this topic.}
%\bibliographystyle{plain}
\bibliography{refs.bib}

\end{document}